\documentclass[11]{llncs}
\usepackage{fullpage}
\usepackage{amsmath}
\usepackage{amssymb}
\usepackage[pdftex]{graphicx}
\begin{document}
\bibliographystyle{splncs}
\title{Memory Consistency Conditions for Self-Assembly Programming}

\author{Aaron Sterling\thanks{This research was supported in part by National Science Foundation Grants 0652569 and 0728806.}}

\institute{Laboratory for Nanoscale Self-Assembly\\Department of Computer Science\\Iowa State University\\Ames, Iowa, USA\\  \email{sterling@cs.iastate.edu}}

\maketitle
\begin{abstract}
Perhaps the two most significant theoretical questions about the programming of self-assembling agents are: (1) necessary and sufficient conditions to produce a unique terminal assembly, and (2) error correction.  We address both questions, by reducing two well-studied models of tile assembly to models of distributed shared memory (DSM), in order to obtain results from the memory consistency systems induced by tile assembly systems when simulated in the DSM setting.  The Abstract Tile Assembly Model (aTAM) can be simulated by a DSM system that obeys causal consistency, and the locally deterministic tile assembly systems in the aTAM correspond exactly to the concurrent-write free programs that simulate tile assembly in such a model.  Thus, the detection of the failure of local determinism (which had formerly been an open problem) reduces to the detection of data races in simulating programs.  Further, the Kinetic Tile Assembly Model can be simulated by a DSM system that obeys GWO, a memory consistency condition defined by Steinke and Nutt.  (To our knowledge, this is the first natural example of a DSM system that obeys GWO, but no stronger consistency condition.)  We combine these results with the observation that self-assembly algorithms are local algorithms, and there exists a fast conversion of deterministic local algorithms into deterministic self-stabilizing algorithms.  This provides an ``immediate'' generalization of a theorem by Soloveichik \emph{et al.} about the existence of tile assembly systems that simultaneously perform two forms of self-stabilization: proofreading and self-healing.  Our reductions and proof techniques can be extended to the programming of self-assembling agents in a variety of media, not just DNA tiles, and not just two-dimensional surfaces.
\end{abstract}
\newpage
\setcounter{page}{1}
\section{Introduction}
In August 2009, IBM and the DNA and Natural Algorithms Group at Caltech announced a joint project to use ``DNA origami'' as a scaffolding in order to place microchip components 6 nm apart, breaking the 22 nm barrier that is the current state of the art in chip manufacturing~\cite{ibm_caltech_origami}.  In contrast to some other emerging models of computation, such as quantum computing or membrane computing, about which there are extensive theoretical results but as yet little experimental success, the experimental results of algorithmic DNA self-assembly are significantly ahead of the theory.  (A recent survey of nanofabrication by DNA self-assembly, including a high-level explanation of DNA origami, appears in~\cite{nanofabrication-survey}.)  Indeed, there are only a handful of results about perhaps the two most significant theoretical questions about self-assembly programming: (1) necessary and sufficient conditions to produce a unique terminal assembly, and (2) error correction.  The goal of this paper is to recast those two questions as programming questions of memory consistency conditions and self-stabilizing algorithms, thus making techniques from the study of concurrent architectures and programming languages, and self-stabilization, available to this emerging area of research.

When a global structure (or organism) forms because of the connections formed by strictly simpler structures to one another, following only local rules, we say the global structure self-assembles.  The goal of \emph{algorithmic self-assembly} is to direct (or to program) the self-assembly of desired structures, by constructing self-assembling agents, and their environment, so they combine to form a desired result.  We will focus on algorithmic \emph{DNA} self-assembly in this paper, a field merging computer science and nanotechnology that began in the 1990s, spurred especially by the work of Adleman, Rothemund, and Winfree~\cite{sticker}.  The formalisms to model self-assembling systems contain the following: a finite set of distinct types of self-assembling agents, a set of local binding rules that completely determines the behavior of the agents, and an initial configuration of the system. A particular self-assembly ``run'' starts with an operator placing a finite seed assembly on the surface, and then allowing a ``solution'' containing infintely many of each agent type to mix on the surface.  Agents bind to the seed assembly, and to the growing configuration, consistent with the local rules, and in a random, asynchronous manner.  In the tile assembly models we consider in this paper, each agent is a four-sided tile, and the assembly surface is the two-dimensional integer plane.

In~\cite{sterling}, we proved a time lower bound for certain computational problems in self-assembly on a surface, by reducing a class of self-assembly models to message-passing models of distributed computing, and then applying a known impossibility result about local algorithms. We follow a similar strategy in this paper: we reduce well-studied tile assembly models to models of distributed shared memory and then explore some of the consequences of those reductions.  In particular, we consider the Abstract Tile Assembly Model (aTAM), due to Winfree~\cite{winfree} and Rothemund~\cite{rothemund}, and the Kinetic Tile Assembly Model (kTAM) due to Winfree~\cite{winfree}.  In the aTAM, binding between self-assembling agents is error-free and irreversible; while in the kTAM, binding errors are possible, and agents can bind but later dissociate with some positive probability.  We show that tile assembly systems defined in the aTAM can be simulated by systems of distributed shared memory (DSM) that obey \emph{causal consistency}~\cite{causalmemory}, a memory consistency condition weaker than the better-known sequential consistency.  In a sense, this level of memory consistency is ``tight'' for any DSM model that simulates the aTAM.

Next, we translate one of the fundamental theorems about the aTAM---that ``locally deterministic'' tile assembly systems produce a unique terminal assembly~\cite{solwin}---into the language of memory consistency.  We show that locally deterministic tile assembly systems correspond exactly to the concurrent-write free programs that simulate tile assembly in our DSM model.  Hence, the programming techniques to produce data-race free and concurrent-write free programs---and to detect data races---can be applied to the programming of self-assembling agents.

Regarding the kTAM, we show it reduces to a model of DSM obeying memory consistency condition GWO, which is strictly weaker than causal consistency.  Again, there is a sense in which this level of memory consistency is tight.  GWO was defined by Steinke and Nutt, to fill out a lattice with which they compared all known memory consistency conditions~\cite{memory_consistency}.  To our knowledge, the only DSM system in the literature that precisely obeys GWO is the one Steinke and Nutt built to show that some such model exists.  The DSM simulation of the kTAM, then, is the first natural example of a model that lies within GWO but no stronger level of memory consistency.

Finally, we combine these results with the observation that self-assembly algorithms are local algorithms, and there exists a fast conversion of deterministic local algorithms into deterministic self-stabilizing algorithms.  This provides an ``immediate'' generalization of a theorem by Soloveichik \emph{et al.} about the existence of tile assembly systems that simultaneously perform two forms of self-stabilization: proofreading and self-healing.  Our general reduction and proof techniques can be extended to the programming of self-assembling agents in a variety of media, not just DNA tiles, and not just two-dimensional surfaces.

Several researchers have voiced intuitions about a connection between self-assembly and distributed computing, for example Klavins~\cite{klavins-csm07}, or an Arora \emph{et al.} 2007 NSF Report~\cite{report}.  However, the first rigorous application of the theory of distributed computing to questions in self-assembly appeared in~\cite{sterling}.  Subsequent to~\cite{sterling} (and its extended version~\cite{surface}), we used the wait-free consensus hierarchy to separate models of self-assembly based on their synchronization power~\cite{datsa}; and we showed that graph assembly systems (a graph grammar self-assembly formalism due to Klavins) are distributed systems in a strong sense~\cite{sa-systems-are-dsystems}. To the best of our knowledge, the current paper is the first to consider self-assembly within the context of memory consistency models, and the first to apply multiprocessor programming techniques to biomolecular computation

The rest of the paper is structured as follows.  Section~\ref{section:background} provides background on tile assembly models and memory consistency models.  In Section~\ref{section:atamreduction} we reduce the Abstract Tile Assembly Model to a DSM system that obeys causal consistency.  In Section~\ref{section:ktamreduction} we reduce the Kinetic Tile Assembly Model to a DSM system that obeys GWO.  In Section~\ref{section:proofreading} we show how to generalize an existence theorem about proofreading and self-healing tilesets by using techniques from self-stabilizing algorithms.  Section~\ref{section:conclusion} concludes the paper and suggests directions for future research.
\section{Background} \label{section:background}
\subsection{Tile assembly background}
We now give the formal definitions of the tile assembly models we will work with.
\subsubsection{Abstract Tile Assembly Model}
Winfree's objective in defining the Abstract Tile Assembly Model was to provide a useful mathematical abstraction of DNA tiles combining in solution in a nondeterministic, asynchronous manner~\cite{winfree}.  Rothemund~\cite{rothemund}, and Rothemund and Winfree~\cite{programsize}, extended the original definition of the model.  For a comprehensive introduction to tile assembly, we refer the reader to~\cite{rothemund}.  Intuitively, we desire a formalism that models the placement of square tiles on the integer plane, one at a time, such that each new tile placed binds to the tiles already there, according to specific rules.  Tiles have four sides (often referred to as north, south, east and west) and exactly one orientation, \emph{i.e.}, they cannot be rotated.

A tile assembly system $\mathcal{T}$ is a 5-tuple $(T,\sigma,\Sigma, \tau, R)$, where $T$ is a finite set of tile types; $\sigma$ is the \emph{seed tile} or \emph{seed assembly}, the ``starting configuration'' for assemblies of $\mathcal{T}$; $\tau:T \times \{N,S,E,W\} \rightarrow \Sigma \times \{0,1,2\}$ is an assignment of symbols (``glue names'') and a ``glue strength'' (0, 1, or 2) to the north, south, east and west sides of each tile; and a symmetric relation $R \subseteq \Sigma \times \Sigma$ that specifies which glues can bind with nonzero strength.  In this model, there are no negative glue strengths, \emph{i.e.}, two tiles cannot repel each other.

A \emph{configuration of $\mathcal{T}$} is a set of tiles, all of which are tile types from $\mathcal{T}$, that have been placed in the plane, and the configuration is \emph{stable} if the binding strength (from $\tau$ and $R$ in $\mathcal{T}$) at every possible cut is at least 2.  An \emph{assembly sequence} is a sequence of single-tile additions to the frontier of the assembly constructed at the previous stage.  Assembly sequences can be finite or infinite in length.  The \emph{result} of assembly sequence $\overrightarrow{\alpha}$ is the union of the tile configurations obtained at every finite stage of $\overrightarrow{\alpha}$.  The \emph{assemblies produced by $\mathcal{T}$} is the set of all stable assemblies that can be built by starting from the seed assembly of $\mathcal{T}$ and legally adding tiles.  If $\alpha$ and $\beta$ are configurations of $\mathcal{T}$, we write $\alpha \longrightarrow \beta$ if there is an assembly sequence that starts at $\alpha$ and produces $\beta$.  An assembly of $\mathcal{T}$ is \emph{terminal} if no tiles can be stably added to it.

We are, of course, interested in being able to \emph{prove} that a certain tile assembly system always achieves a certain output.  In~\cite{solwin}, Soloveichik and Winfree presented a strong technique for this: local determinism.  An assembly sequence $\overrightarrow{\alpha}$ is \emph{locally deterministic} if (1) each tile added in $\overrightarrow{\alpha}$ binds with the minimum strength required for binding; (2) if there is a tile of type $t_0$ at location $l$ in the result of $\alpha$, and $t_0$ and the immediate ``OUT-neighbors'' of $t_0$ are deleted from the result of $\overrightarrow{\alpha}$, then no other tile type in $\mathcal{T}$ can legally bind at $l$; the result of $\overrightarrow{\alpha}$ is terminal.  Local determinism is important because of the following result.
\begin{theorem}[Soloveichik and Winfree~\cite{solwin}]
If $\mathcal{T}$ is locally deterministic, then $\mathcal{T}$ has a unique terminal assembly.
\end{theorem}
\subsubsection{Kinetic Tile Assembly Model}
The Kinetic Tile Assembly Model (kTAM) was defined by Winfree~\cite{winfree}, to provide a mathematical model for self-assembly (and disassembly) in solution, based on the kinetics of chemical reactions.  Slightly different versions of the kTAM appear in different papers on the subject.  We will follow the treatment in~\cite{selfhealing}, because we will use techniques from distributed computing to address an open question in~\cite{selfhealing}.

Whereas the aTAM is an error-free, irreversible, nondeterministic model, the kTAM is a probabilistic model in which tiles bind with some probability of error, and bound tiles can dissociate with some probability.  These probabilities are derived from the equations of chemical reaction kinetics.  There is a \emph{forward rate} $f$, which we assume is the same for any tile type at any position of the perimeter in the growing assembly, defined as $f=k_fe^{-G_{mc}}$, where $k_f$ is a constant that sets the time scale, and $G_{mc}$ is the logarithm of the concentration of each tile type in solution.  We assume that tiles can only fall off of the perimeter of the assembly; this assumption matches experimental observation.  The rate of dissociation (\emph{reverse rate} $r_b$) depends exponentially on the number of bonds that must be broken: $r_b=k_fe^{-bG_{se}}$, where $b$ is the total interaction strength with which the tile is attached to the assembly, and $G_{se}$ is the unit bond free energy, which may depend on the overall temperature of the system.

As with our treatment of the aTAM, we assume that ``strength 2'' bonds are sufficient for tiles in the kTAM to bind stably.  Hence we let $f=r_2$, which ensures that the forward growth of the kTAM mirrors (with high probability) the binding rules of the aTAM, and incorrectly bound tiles (with high probability) quickly dissociate.  We assume that $k_f$ is a physical constant that cannot be experimentally controlled, but by changing concentrations or temperature we could change $G_{mc}$ and/or $G_{se}$.

One objective of this paper is to develop proof techniques that would be applicable to models of self-assembly other than the aTAM and kTAM.  To this end, we will not use all the information available about the kTAM when we simulate it, but rather will use a more general construction, with the understanding that specific values for probabilities could be plugged in as required, based on the rate equations of the kTAM.  In particular, we will limit ourselves to the existence of a forward rate $f$ which is the same for each bond, a reverse rate $r_b$ which is much higher for erroneous bonds than for correct ones, and an error probability $\pi_e$ that a tile will bind incorrectly to the frontier of the assembly.
\subsection{Distributed shared memory background}
A \emph{distributed shared memory} model (or system) is a model of distributed processors and (possibly shared) read/write registers.  A processor $p$ can perform a \emph{read} or a \emph{write} on register $r$, if $p$ has permission to perform that \emph{operation} on $r$.  The only operations a processor can perform on a register are reads and writes.  The read operation begins with an \emph{invocation} and terminates when $p$ receives a \emph{value}.  The write operation begins with an invocation that includes a value, and ends when $p$ receives an \emph{ack}.  A processor can only perform one operation at a time.  We do not assume atomicity of reads and writes to a given register $r$.

A \emph{memory consistency model} specifies the allowable behavior of memory.  Study of memory consistency models arose from a conflict between the goals of hardware and compiler designers, to permit aggressive optimization (which requires ``weak'' memory consistency), and the desire of programmers to have concurrent code execute in a predictable fashion (which requires ``strong'' memory consistency).  We refer the reader to~\cite{mem-consist-tutorial} for a survey and tutorial on these issues.  In 2004, Steinke and Nutt presented a theory that unified the various memory consistency models that had been proposed in the literature~\cite{memory_consistency}.  They showed the existence of a lattice of 13 memory consistency models; this lattice contained all known models, and showed the logical interrelation between each.  We now describe the memory consistency conditions that will be most important for this paper.

A system of distributed processors is \emph{sequentially consistent}~\cite{sequential-consistency} if the result of any execution is the same as if the operations of all the processors were executed in some sequential order, and the operations of each individual processor appear in this sequence in the order specified by its program.  A system of distributed processors is \emph{causally consistent}~\cite{causalmemory} if for each processor the operations of that processor plus all writes known to that processor appear to that processor to appear in a total order that respects potential causality.  A system of distributed processors is \emph{PRAM consistent}~\cite{pram-consistency} if writes performed by a single processor are seen by all other processors in the order in which they were issued, but writes from different processors may be seen in different orders by different processors.  These three consistency conditions are in descending order of strength: a sequentially consistent system is causally consistent, and a causally consistent system is PRAM consistent, but the converses do not always hold.

To compare consistency conditions, Steinke and Nutt defined logical properties about \emph{processor views} of a DSM system.  A processor view is a total order on a subset of operations that occurred during an execution of a DSM system, the subset being those operations an individual processor performed or could deduce occurred.  The property \emph{GPO} (``global process order'') is the condition that there is global agreement on the order of operations from each process.  The property \emph{GWO} (``global write-read-write order'') is the condition that there is global agreement on the order of any two writes when a process can prove it has read one before the other. A DSM system satisfies GPO exactly when it is PRAM consistent, and it satisfies both GPO and GWO exactly when it is causally consistent.

A relation $\prec$ is a \emph{causality order} of operations if $o_1 \prec o_2$ means that one of the following holds for any operations $o_1$ and $o_2$:
\begin{enumerate}
\item $o_1$ and $o_2$ were performed by the same processor $p$, and $p$ executed $o_1$ before $o_2$.
\item $o_2$ reads the value written to a shared register by $o_1$.
\item There is some other operation $o'$ such that $o_1 \prec o' \prec o_2$.
\end{enumerate}
We will say that program $P$ is \emph{concurrent-write free} in DSM model $\mathcal{M}$ if every processor in $\mathcal{M}$ runs $P$, and there is no legal execution history $H$ of $P$ such that the causality order induced by $H$ contains two writes $w_1$ and $w_2$ such that $w_1 \nprec w_2$ and $w_2 \nprec w_1$ (\emph{i.e.}, there are no writes that are causally concurrent under any possible execution).  In other words, $P$ is concurrent-write free if a single program $Q$ that simulates each processor in the system executing $P$ is concurrent-write free under the standard definition that no execution of $Q$ contains conflicting writes.
\section{Reduction of tile assembly models to distributed shared memory models}
\subsection{The aTAM reduces to causally consistent models of DSM} \label{section:atamreduction}
The objective of this section is to show that the aTAM can be simulated by a causally consistent DSM system, and, under a reasonable definition of ``simulation,'' no DSM system that fails to obey causal consistency can simulate the aTAM.  First, we define formally what it means for a DSM system to simulate a tile assembly system.  For simplicity, we limit consideration to tile assembly systems that self-assemble on the first quadrant of $\mathbb{Z}^2$; our definitions could be extended to the entire integer plane.
\begin{definition}
Let $\mathcal{T}=\langle T,\sigma,\Sigma,\tau,R\rangle$ be a tile assembly system in the aTAM, and $M$ be a model of distributed processors.  We say $M$ \emph{simulates $\mathcal{T}$ on the $k \times k$ surface} if the following holds.
\begin{enumerate}
\item The network topology of $M$ is that of a $k \times k$ square.
\item Each processor $p_{ij}$ $(0 \leq i,j < k)$ in $M$ has a write-once output register $\rho_{ij}$ that is initialized to value EMPTY, and can take on $|T|$ distinct values other than EMPTY.
\item There is a bijection $\psi$ from tile types in $T$ to possible values of $\rho$ (not including EMPTY).
\item For each legal tile assembly sequence $\overrightarrow{\alpha}$ of $\mathcal{T}$, there is a legal execution $E$ of $M$ such that, if tile type $t$ is placed on coordinate $(i,j)$ in $\overrightarrow{\alpha}$, then processor $p_{ij}$ writes $\psi(t)$ to $\rho_{ij}$ in $E$.  Moreover, processors write to their respective $\rho$ in $E$ in the same order that tiles get placed on the surface in $\overrightarrow{\alpha}$.  (The placement of the seed assembly $\sigma$ is simulated by writing the value $\psi(t)$ to $\rho_{ij}$ if $\sigma(i,j)=t$.)
\item For each legal execution $E$ of $M$, there is a corresponding legal assembly sequence $\overrightarrow{\alpha}$, such that if $p_{ij}$ writes $\psi(t)$ to $\rho_{ij}$, then tile type $t$ is placed on location $(i,j)$ in $\overrightarrow{\alpha}$.  Moreover, the order of writing values in $E$ is preserved by the order of placing tiles in $\overrightarrow{\alpha}$.
\end{enumerate}
\end{definition}
Intuitively, $M$ simulates $T$ on a finite surface if each process in $M$ behaves like a location on the surface, with each processor executing a local algorithm that mimics the binding rules required by $R$.  We now generalize the above definition to arbitrary tile assembly systems.
\begin{definition}
Let $\mathcal{T}$ be a tile assembly system in the aTAM.  We say a class $\mathcal{M}=\{M_0,M_1,\ldots\}$ of DSM models \emph{simulates} $\mathcal{T}$ if, for each $k \in \mathbb{N}$, $M_k$ simulates $\mathcal{T}$ on a $k \times k$ surface. Let $\phi$ be a mapping from tile assembly systems to algorithms.  Then we say that $(\mathcal{M},\phi)$ \emph{simulates} the aTAM if, for any tile assembly system $\mathcal{T}$, $\mathcal{M}$ simulates $\mathcal{T}$ when the processors in the models in $\mathcal{M}$ run $\phi(\mathcal{T})$ as their local algorithm, beginning at an initial state determined by the seed assembly $\sigma$ as above.
\end{definition}
We now prove that there exists a class of causally consistent DSM models that simulates the aTAM.
\begin{theorem} \label{theorem:aTAMsimulation}
There exists a class of DSM models $\mathcal{M}$ that simulates the aTAM.  Each $M \in \mathcal{M}$ is causally consistent.  Further, the models in $\mathcal{M}$ do not obey a memory consistency condition in Steinke and Nutt's lattice that is stronger than causal consistency.
\end{theorem}
\begin{proof}
Fix $k \in \mathbb{N}$.  We define a DSM model $M_k$ as follows.  $M$ contains $k^2$ processors, with network graph of a $k \times k$ grid.  We will refer to the processors as $p_{ij}$ $(i,j < k)$, to denote the processor at location $(i,j)$ in the network grid.  Note this is a convenience for the proof; the processors do not have unique ID's, and do not know whether they are on the edge, the corner, or the interior of the grid.  We assume that all tile assembly systems have temperature 2, as that is sufficient for Turing universality.  Each processor can read from two registers, $r^1_{ij}$ and $r^2_{ij}$.  Each processor can write to twelve registers, $r^1$, $r^2$ and $Index$ of each of its neighbors to the north, south, east and west on the grid.  (For processors on the edge of the grid, these registers exist, and ack when written to, but no processor ever reads from them if they ``belong'' to nonexistent processors.  This way, a processor cannot deduce that it is on the edge of the grid.)  To simulate $\tau>2$, we could use registers $r^1,\ldots,r^{\tau}$ instead of just two such registers per processor.

The register $Index_{ij}$ is initialized to the value 1, and can take 3 possible values: 1, 2, and ``1 and 2.''  The algorithm each processor runs will look first at the value in $Index$ to determine whether to write to $r_1$, $r_2$ or both.

Each processor $p_{ij}$ has a write-once register $\rho_{ij}$ initialized to the value EMPTY.  Only $p_{ij}$ can write to $\rho_{ij}$, and it can write one of $|T|$ distinct values.  Fix a bijection $\psi$ between tile types of $T$ and possible values that can be written to each $\rho$.  The processors of $M_k$ all run a common nondeterministic, distributed, local, algorithm as follows.  Before starting execution, $M_k$ is configured to simulate the seed assembly of $\mathcal{T}$: for $i,j<k$, if $\sigma(i,j) \neq \emptyset$ then $p_{ij}$ writes $\psi(\sigma(i,j))$ to $\rho_{ij}$.

Once placement of the seed assembly is simulated, execution of the algorithm proceeds in synchronized stages (rounds), beginning at stage $s=0$.  At the start of time stage $s$, if at stage $s-1$, $p_{ij}$ wrote a value $\psi(t)$ (for $t \in T$) to $\rho_{ij}$, then $p_{ij}$ writes to each of its neighbors in a way that communicates the glues and glue strengths of the tile $p_{ij}$ is simulating, as follows.  (We limit discussion to communication with the neighbor to the north; communication to the other neighbors is similar.)

First, $p_{ij}$ reads $Index_{ij+1}$.  If $p_{ij}$ wants to communicate a bond with strength 2 to its northern neighbor, then it writes its message in both $r^1_{ij+1}$ and $r^2_{ij+1}$.  If it wants to communicate a bond with strength 1, then $p_{ij}$ writes $\langle \textrm{S},g \rangle$ to $r^1_{ij+1}$, to $r^2_{ij+1}$, or to both registers, based on the value it read from $Index_{ij+1}$; here $g$ is a message that corresponds to the glue type on the north side of $t$.  This indicates to the northern neighbor of $p_{ij}$ that glue type $g$ is present immediately to the south.  More generally, each message written to $r^1$ or $r^2$ of a given processor will be of form $\langle d, g \rangle$ where $d \in \{\textrm{N, S, E, W} \}$ is a direction, and $g$ is a glue type.

After processors write, the second phase of stage $s$ takes place as follows.  The algorithm of exactly one processor $p_{ij}$ (chosen nondeterministically from the set of processors that could legally write a value to their respective $\rho$ using the protocol explained below) with $p_{ij}$ still set to EMPTY chooses (again nondeterministically) which message it will ``hear'' in each of $r^1_{ij}$ and $r^2_{ij}$, of all the messages that have been written to $r^1_{ij}$ and $r^2_{ij}$ since the start of execution of the algorithm.  Processor $p_{ij}$ then writes a value to $\rho_{ij}$, by applying $\psi$ and $\psi^{-1}$ to the binding rules induced by the relation $R$ of $\mathcal{T}$, so $p_{ij}$ writes $\psi(t)$ for some $t$ such that $t$ can legally bind given neighbors with glue types as indicated by the values of $r^1_{ij}$ and $r^2_{ij}$.  If more that one tile type can legally bind given the same set of neighbors, the value $\psi(t)$ is again chosen nondeterministically from the set of legal values.

Finally, $p_{ij}$ ``increments'' the value of $Index$ of each of its neighbors that it wrote to.  Continuing with our example of writing to the northern neighbor, if the value is 1, it writes the value 2 to $Index_{ij+1}$.  If the value is 2, it writes ``1 and 2'' to $Index_{ij+1}$.  This concludes stage $s$ of the algorithm.

$M_k$ simulates $\mathcal{T}$ on a $k \times k$ surface in the ``natural'' way.  For each nondeterministic run of $\mathcal{T}$, \emph{i.e.}, each tile assembly sequence $\overrightarrow{\alpha}$, there is a nondeterministic execution of $M_k$ such that the behavior of each $p_{ij}$ mimics the behavior of the location $(i,j)$ in $\overrightarrow{\alpha}$.  Similarly, for each nondeterministic execution $E$ of $M_k$, there is a legal tile assembly sequence that makes the same nondeterministic choices, since the choices of $E$ are constrained by the binding rules that determine the legal behavior of any assembly sequence of $\mathcal{T}$.

Since our choice of $k$ was arbitrary, we can define $\mathcal{M} = \{ M_k \mid k \in \mathbb{N} \}$.  Such an $\mathcal{M}$ simulates $\mathcal{T}$.  Further, our definition of each $M_k$ was uniform with respect to any set of binding rules determined by a given $R$.  So we can define a class of models $\mathcal{M}'$ such that each $M'_k$ runs an algorithm that simulates $\sigma$ and $R$ for any tile assembly system.  Hence there is a class of models $\mathcal{M}'$ that simulates the aTAM.

We turn now to the memory consistency conditions obeyed by the $M'_k$.  First, each $M'_k$ must be causally consistent, because each processor writes only once to any memory location, and writes deterministically based on values read from other writes (values written either to $r^1$ and $r^2$), or at stage 0 to simulate the seed assembly, which is a finite, completely determined set of decisions.  Therefore, for any process $p_{ij}$, the operations of that process, and any writes known to that process, occur in a total order that respects potential causality, even though the values of those writes were nondeterministically chosen.

Causal consistency is an upper bound for the memory consistency of the $M'_k$ as well.  This is because the only memory consistency conditions that are stronger than causal consistency in the lattice by Steinke and Nutt are conditions that include the property GDO (Global Data Order, which is equivalent to cache consistency).  But $M'_k$ makes no guarantee that all writes to the same memory location are performed in some sequential order: a write to $r^1$ or $r^2$ that occurs at an earlier stage than another write may still be nondeterministically chosen as the value of that register in a legal execution.  So causal consistency exactly captures the memory consistency of this simulation of the aTAM.
\end{proof}
It is worth noting that if a DSM model is going to simulate the behavior of a tile assembly system in the aTAM, then causal consistency is the weakest memory consistency model it can follow.  This is so because causal consistency is the combination of properties GPO and GWO.  If a DSM model does not satisfy GPO (\emph{i.e.}, does not satisfy PRAM consistency), then consider processor $p_{ij}$ that writes multiple times to a neighbor $q$, yet $q$ does not see these writes in the order in which they were issued.  There is no legal tile assembly sequence that corresponds to such an execution, as tiles in the aTAM are placed one at a time, error-free, based on information transmitted to neighboring locations by already-placed tiles.  Similarly, if a DSM model does not satisfy GWO, there exists some execution where $q_1$ and $q_2$ disagree on the ordering of writes, even though $q_1$ can prove that a particular write happened first. As before, no legal tile assembly sequence captures this behavior, as a tile assembly sequence induces a total order on the system, such that at each stage of assembly, a newly placed tile communicates its glue types to all neighboring locations.  So with respect to the lattice defined by Steinke and Nutt, causal consistency is ``tight'' for simulations of the aTAM.

We now show that, under our reduction, locally deterministic tile assembly systems correspond to an important class of simulating programs.  Let $P$ be a program that simulates a tile assembly system; we will call $P$ \emph{binding-rule determined} if for each set of messages that simulates strength 2 bonds, there is at most one value $P$ writes to $\rho$ to simulate a tile type.  It is easy to check whether $P$ is binding-rule determined, but harder to check whether a tile assembly system is locally deterministic, as we discuss below.
\begin{theorem} \label{theorem:local-det-is-cwf}
Let $\mathcal{M}$ be the class of DSM systems in Theorem~\ref{theorem:aTAMsimulation}.  Then if $\mathcal{T}$ is a locally determinstic assembly system, the program $P$ for which $\langle \mathcal{M},P \rangle$ simulates $\mathcal{T}$ is concurrent-write free.  Conversely, if $\mathcal{M}$ running concurrent-write free, binding-rule determined, program $P$, simulates tile assembly system $\mathcal{T}$, then $\mathcal{T}$ is locally deterministic.
\end{theorem}
\begin{proof}
Let $\mathcal{T}$ be a locally deterministic tile assembly system, let $\overrightarrow{\alpha}$ be a legal assembly sequence of $\mathcal{T}$, and let $(x,y)$ be a location in the result of $\overrightarrow{\alpha}$.  As before, we assume we are operating within the aTAM at temperature 2.  Either $(x,y)$ is part of the seed assembly, or, in the sequence $\overrightarrow{\alpha}$, prior to a tile being placed at $(x,y)$, either one neighbor of $(x,y)$ has a tile placed with a strength 2 bond incident to $(x,y)$, or two neighbors have tiles placed with strength 1 bonds incident to $(x,y)$.  The simulation of $\overrightarrow{\alpha}$ in $\mathcal{M}$ via the reduction of Theorem~\ref{theorem:aTAMsimulation} produces a concurrent-write free program, as for processor $p_{xy}$, each of $r^1_{xy}$ and $r^2_{xy}$ is written to at most once; and, if two neighbors of $p_{xy}$ write to $Index_{xy}$ with writes $w_1$ and $w_2$, in all legal execution histories, there will be an operation of one processor reading the value written by the other processor to $Index$ in between $w_1$ and $w_2$, so either $w_1 \prec w_2$ or $w_2 \prec w_1$, because of the definition of causal ordering.

Conversely, suppose $\langle \mathcal{M},P \rangle$ simulates a tile assembly system $\mathcal{T}$, and $P$ is concurrent-write free. Then there is no execution history for $\langle \mathcal{M},P \rangle$ that contains concurrent writes on $r^1_{ij}$ and $r^2_{ij}$ for any $p_{ij}$.  But this means that at most two neighbors could have written to the registers of $p_{jk}$, as otherwise two neighbors $q_1$ and $q_2$ would have written to $r^1_{ij}$, and causal consistency permits histories in which those writes could happen in either order.  So the writes are concurrent after all, contrary to assumption. So at most two neighbors write to any location before that location decides which tile type to simulate.  Further, suppose one of the neighbors of $p_{ij}$ that writes to the registers of $p_{ij}$ before $p_{ij}$ decides, writes to $p_{ij}$ with a strength 2 bond.  Then that neighbor writes to both $r^1_{ij}$ and $r^2_{ij}$, so no other neighbor can write to the registers of $p_{ij}$ before $p_{ij}$ decides, or there will be concurrent writes, by the above argument.  Hence the neighbors writing to $p_{ij}$ write messages that simulate exactly a strength 2 bond.  Finally, since $P$ is binding-rule determined, the conditions of local determinism are satisfied, and $\mathcal{T}$ must be locally deterministic.
\end{proof}
%
``$\mathcal{T}$ is locally deterministic'' is an undecidable property, as the standard tile assembly Turing machine simulation is locally deterministic, and it could be modified to do something not locally deterministic iff a machine achieves a halting state.  Nevertheless, it would be useful to test for that property when programming---and debugging---a tile assembly system, hence self-assembly simulation and programming tools have attempted to include that functionality~\cite{isutas}~\cite{tiletemplates}.  Theorem~\ref{theorem:local-det-is-cwf} classifies this problem precisely, and indicates that programming techniques to ensure~\cite{typed-data-race-freedom}, and detect~\cite{data-race-detection}, data-race freedom and concurrent-write freedom in parallel systems can be used productively to program self-assembling agents.
\subsection{The kTAM reduces to GWO-consistent models of DSM} \label{section:ktamreduction}
Whereas the aTAM was a nondeterministic model, the kTAM is a stochastic model.  In order to apply a result about self-stabilizing algorithms to the kTAM, we will construct a DSM simulation whose processors run deterministic algorithms and whose registers return values probabilistically, for example by providing erroneous information with some nonzero probability.  We use the same definition as above for what it means for a DSM system to simulate a tile assembly system, except substituting ``kTAM'' for ``aTAM'' when it appears, and the $\rho_{ij}$ can be written to multiple times, not just once, since the kTAM is a reversible model in which binding errors can occur.  It turns out that the probabilistic behavior of the registers is captured by the memory consistency condition GWO.
\begin{theorem} \label{theorem:kTAMreduction}
There exists a class of DSM models $\mathcal{M}$ that simulates the kTAM.  Each $M_k \in \mathcal{M}$ obeys GWO.  Further, the models in $\mathcal{M}$ do not obey a memory consistency condition in Steinke and Nutt's lattuce that is stronger than GWO.
\end{theorem}
\begin{proof}
We use a DSM system that is largely the same as the one used in Theorem~\ref{theorem:aTAMsimulation}: each processor has an $Index$ register, registers $r^1$, $r^2$, and $\rho$.  $Index$ and $\rho$ when read always returns the most recent value written to them, but $r^1$ and $r^2$ only satisfy a weaker consistency condition, as follows.

Let $V = \{\textrm{EMPTY}, v_1, v_2, \ldots, v_i \}$ be the set of values that has been written to $r^1$ (the behavior of $r^2$ will follow this same condition also).  Let $\pi_e$ be the probability of a tile binding error in the kTAM.  Then, when $r^1$ is read, with probability $\pi_e$ it returns a value that simulates a glue in the tile assembly system, but is not a value that had been previously written to $r^1$.  With probability $1-\pi_e$, $r^1$ returns one of the values previously written to it.  This value is determined by selecting at random from the sample space $V$, with each of the $v_i$ weighted by $f$.  Lastly, if in a previous stage, $r^1$ returned an erroneous value, at subsequent stages, with probability $r_b$, $r^1$ when read will return the value EMPTY; this simulates the dissociation of a bond.  After returning the value EMPTY, $r^1$ no longer returns EMPTY with probability $r_b$, unless at some future stage it returns an erroneous value again.

Execution of the algorithm proceeds in synchronous rounds, much as before.  This time, though, instead of one processor nondeterministically choosing to be the location to act this round, all processors on the perimeter of the simulated assembly act if they can.  (Recall that the kTAM assumes, consistent with experimental observation, that addition and dissociation of tiles only occurs on the perimeter; we mimic that assumption here.)  The specific algorithm is as follows.

At stage 0, each $p_{ij}$ that is part of the seed assembly writes the appropriate value to $\rho_{ij}$.  All other $\rho$, and all $r^1$ and $r^2$ are initialized to EMPTY.  All values of $Index$ are set to 1.

At stage $s \geq 1$, all $p_{ij}$ with non-EMPTY value in $\rho_{ij}$ write glues to their neighbors following the same method as in the aTAM simulation (\emph{i.e.}, reading from $Index$ and using that value to determine which subset of $\{r^1,r^2\}$ to write to).  Then, each processor on the perimeter of the assembly reads the contents of its $r^1$ and $r^2$; those registers return values probabilistically, as explained above.  Each processor on the perimeter then writes to $\rho$, if appropriate (based on the binding rules of $\mathcal{T}$), perhaps to EMPTY.  To conclude the round, processors that sent messages at the beginning of the round increment the value of the appropriate $Index$ registers, similar to the aTAM simulation.  

The behavior of each $p_{ij}$ mirrors the behavior of the locations of the surface on which $\mathcal{T}$ is assembling.  So for any legal execution of $\mathcal{M}$, there is a legal tile assembly sequence in which the tiles were placed in the same order.  Note that there is a significant change from the previous simulation:  multiple processors may act in a round, instead of just one at a time, as in the aTAM simulation, so multiple locations may write to the same $r^1$.  That is consistent with our definition of ``simulation,'' because for every tile assembly sequence of $\mathcal{T}$, there will be an execution of $\mathcal{M}$ that writes to each $\rho$ in the same order that locations add or remove tiles.  The order is what matters, not the exact time step at which a change takes place.  So $\mathcal{M}$ simulates the kTAM.

We now show that $\mathcal{M}$ obeys GWO, but no stronger consistency model in the lattice of Steinke and Nutt.  Recall that GWO means that there is global agreement on the order of potentially causally-related writes.  The writes of processor $p$ are causally related in history $H$ to the writes of processor $q$ only if one processor simulates a tile that binds or dissociates, and there is a sequence of processors (WLOG from $p$ to $q$) $S = \langle p,p_1,p_2,\ldots,q \rangle$, such that each processor is a neighbor of its successor, and each processor chose to simulate its particular tile type because of information written to it by its predecessor.  There is global agreement on the ordering of writes, because any processor $q'$ that could view a write $w$ by $p'$, or view a write causally related to $w$, is a member of a sequence $S'$ from $p'$ to $q'$ as above.  So $\mathcal{M}$ satisfies GWO.

On the other hand, $\mathcal{M}$ does not satisfy either of the consistency conditions immediately stronger than GWO in Steinke and Nutt's lattice.  These two stronger conditions are causal consistency (GPO+GWO), and GDO+GWO, where GDO is equivalent to cache consistency.  $\mathcal{M}$ does not satisfy GDO for the same reason as in the aTAM simulation: there is no guarantee that writes performed to a given memory location will return the most recent value as a read.  GPO, which is equivalent to PRAM consistency, was satisfied by the DSM models simulating the aTAM, but is not satisfied under this model because of the possibility that a register will return a value that has never been written to it.  This is in fact a violation of \emph{slow consistency}, which requires that a read return a value that has previously been written to it.  Slow consistency is strictly weaker than PRAM consistency, hence PRAM consistency cannot be satisfied.  So GWO is the strongest consistency condition in the lattice obeyed by $\mathcal{M}$.
\end{proof}
As in the previous section, there is a sense in which GWO is the ``tight'' memory consistency condition for kTAM simulation.  The only consistency condition in Steinke and Nutt's lattice weaker than GWO is \emph{local consistency}, which requires that each process's local operations appear to occur in the order specified by its program.  Without such a requirement, processors in a simulation could potentially change the value of $\rho$ before receiving the values written to them by their neighbors, which would be inconsistent with our intuition of simulating tiles binding to other tiles.  However, the lack of consistency conditions between local consistency and GWO is due only to the fact that no one has defined and studied such conditions, not because it is logically impossible to do so.  More precise consistency conditions for simulation of error-permitting self-assembly may be an area of future investigation.
\section{Self-stabilization applied to tile self-assembly} \label{section:proofreading}
Now that we have reduced the kTAM to models of DSM that run deterministic algorithms with sometimes-faulty registers, we can apply the theory of self-stabilization to prove the existence of certain error-correcting tile assembly systems.  Recall that a distributed system is \emph{self-stabilizing} if, starting from any initial state, the system is guaranteed to converge to, and stay in thereafter, one of a set of ``legitimate'' states; this research topic was begun by Dijkstra~\cite{Dijkstra-self-stabilizing}.  In the case of tile assembly, the legitimate states are the assemblies achievable from error-free assembly sequences.

While not in wide use in the distributed computing literature, there is a simple, polynomial-time conversion that, given a constant-time local algorithm, yields a constant-time self-stabilizing algorithm~\cite{AwerbuchSipser-self-stab}~\cite{AwerbuchVarghese-self-stab}.  (See~\cite{Lenzen2009Local} for a recent exposition of this conversion, with additional motivation and examples.)  Molecular self-assembly algorithms are inherently local, and the self-assembly literature has considered two main classes of self-stabilizing algorithms: \emph{proofreading} tilesets~\cite{proofreading}, which correct initial binding errors; and \emph{self-healing} tilesets~\cite{winfree-selfhealing}, which rebuild completed assemblies that have been damaged.  Soloveichik \emph{et. al.} recently demonstrated a ``proof of concept,'' by constructing a tile assembly system that combined both proofreading and self-healing properties within a restricted version of the kTAM~\cite{selfhealing}.  Their construction worked for tile assembly systems that only built north and east.  We can use tools from self-stabilization to generalize their result to the full kTAM, and to locally deterministic models that grow in any direction.

If $\mathcal{T}$ is a tile assembly system, the \emph{$c$-scaled result of $\mathcal{T}$} is the colored shape on the integer plane obtained by ``blowing up'' each location in the result of $\mathcal{T}$ to a $c \times c$ block of tiles, such that each tile in the block is colored the same as the tile on the source location in the result of $\mathcal{T}$.
\begin{theorem} \label{theorem:self-stab}
There is a polynomial-time algorithm that does the following: upon input of a locally deterministic tile assembly system $\mathcal{T}$ for the kTAM, it outputs a self-healing, proofreading, tile assembly system $\mathcal{T}'$ such that $\mathcal{T}'$ builds the $c$-scaled result of $\mathcal{T}$, for some constant $c$.  Further, $|T'| \leq \binom{4}{3}^2c^2|T|^2$.
\end{theorem}
\begin{proof}
Let $\mathcal{T}$ be a locally deterministic tile assembly system, and let $\langle \mathcal{M}, P \rangle$ be the DSM model that simulates it, as produced by the reduction in Theorem~\ref{theorem:kTAMreduction}.  Then let $P^*$ be the self-stabilizing algorithm obtained by applying the conversion of~\cite{AwerbuchSipser-self-stab}~\cite{AwerbuchVarghese-self-stab} to $P$.  (Briefly, $P$ is a constant-time algorithm.  To convert $P$ to a self-stabilizing algorithm, ``unroll'' all possible executions of $P$, over all possible inputs, as a circuit.  The program $P^*$ simulates that circuit, and assumes the inputs to the circuit are correct.  Then it repeats that same step $k$ times, where $k$ is large enough such that, with high probability, all the inputs at round $k$ really \emph{are} correct.)  $P^*$ is also a constant-time algorithm, whose running time depends on the size of the $P$-simulating circuit and the value of $k$ we choose.

As tile assembly is Turing universal, we can of course convert $P^*$ into tiles.  More importantly, there is some constant $c$ such that we need lay out only $c$-many tiles in order to simulate the behavior $P^*$ on any of its legal inputs.  WLOG we assume that $c$ is large enough, and our tile simulation is ``padded'' if necessary, so that the simulation of $P^*$, on input of tile type $t$, takes up a $c \times c$ square, for any $t$.  Further, we dedicate a set of tile types to each input $t$, so the color of each tile in this $c \times c$ square is the same color as the input tile $t$.  Finally, in order to ensure that self-healing does not generate nondeterministic behavior, we have to differentiate each tile that could potentially add tiles in more than one direction.  (This is why the construction in~\cite{selfhealing} assumed that tiles could only grow north and east.)  In other words, if a ``hole gets punched'' in a completed assembly, self-healing tiles can rebuild the empty area, by building in the reverse direction from the original assembly sequence.  So for each tile type in $\mathcal{T}$, we need $\mathcal{T}'$ to include a unique tile type that encodes ``Tile type $t$ attached to the assembly using input sides $S$.''  It takes at least one side to attach, so at most three sides of $t$ remain as output sides, which is where we need to encode the information to reverse the process of binding $t$ at that location.  So the number of tiles we need to place into $\mathcal{T}'$ to simulate a given tile $t$ of $\mathcal{T}$ is upper-bounded by $\binom{4}{3}$.  Further, we need sufficient new tile types that encode ``I have received information from tile type $t$.''  There are at most another $\binom{4}{3}|T|$ of those.  Since for each tile type in $\mathcal{T}$ there are at most $c^2$ tile types that simulate $P^*$ on $t$, we get the overall upper-bound $|T'| \leq \binom{4}{3}^2c^2|T|^2$.
\end{proof}
A discussion of theoretical and practical reasons to choose particular values of $k$ (\emph{i.e.}, the optimal number of iterations to minimize the possibility of error) for DNA self-assembly appears in~\cite{selfhealing}.  For our purposes, it suffices to use the conversion from constant-time algorithm to self-stabilizing algorithm as a black box, without considering the specific types of error that are most likely to occur in the kTAM.  Therefore, Theorem~\ref{theorem:self-stab} gives much weaker tile complexity bounds than the dedicated construction that appears in~\cite{selfhealing}, but it provides a general method that can be extended to a variety of self-assembly models, not just DNA tiling, and not just two-dimensional surfaces.

Theorem~\ref{theorem:self-stab} may be a ``proof of applicability'' of self-stabilization techniques to self-assembly, but more is needed than ``just'' self-stabilization.  Perhaps the most troublesome form of error in nanoscale self-assembly experiments occurs when a tile binds incorrectly, and then other tiles bind around it, preventing it from dissociating.  To address this in the language of self-stabilization, one needs a \emph{fault-containing} self-stabilizing algorithm with minimal \emph{fault gap}~\cite{fault-containing}, that is to say, a self-stabilizing algorithm in which the effects of a processor's failure are contained within that processor's local neighborhood, and, after recovering from a given fault, there is only a small time gap until the system can recover from a new fault.  Fujibiyashi \emph{et al.} have suggested extending the kTAM with special tile mechanisms that would achieve a one-time-step fault gap with high probability~\cite{error-suppression}, though they did not phrase their results in the language of self-stabilization.  There is some evidence that handling fault-containment probabilistically instead of deterministically will reduce an algorithm's fault-gap~\cite{probabilistic-fault-containment}, but little is known about such tradeoffs, either in self-assembly or in general distributed systems.
\section{Conclusion} \label{section:conclusion}
In this paper, we reduced the Abstract Tile Assembly Model and the Kinetic Tile Assembly Model to systems of distributed shared memory with particular memory consistency conditions.  We then applied the reductions to show that (1) local determinism is closely related to concurrent-write freedom in parallel programming, and (2) the theory of self-stabilization can be usefully applied to questions of error correction in self-assembly.

We focused on the aTAM and the kTAM because they have been the most theoretically studied self-assembly models.  However, both models are limited to the binding of DNA tiles to other DNA tiles, and as the recent nanofabrication survey~\cite{nanofabrication-survey} points out, the ``greatest promise'' of algorithmic DNA self-assembly ``may lie in applications where DNA nanostructure templates have been used to assemble other inorganic components and functional groups.''  The IBM/Caltech microchip project is an example of this research direction.  Therefore, we believe it is critical to develop a programming theory for ``mixed-media'' models of self-assembly (such models by-and-large do not yet exist), and that programming theory may be advanced by continuing the investigation begun in this paper.

From the perspective of ``pure theory,'' there has been initial work to classify shared read/write variables~\cite{shared-variables}, much as Steinke and Nutt classified known memory consistency systems.  It would be useful to explore further the weak consistency, like GWO, offered by registers that simulate the binding of self-assembling agents, whether in DNA or another medium.  It would also be useful to explore what ``more robust'' registers (like the consensus objects in~\cite{datsa}) could be built, to know how agents might cooperate to form more fault-tolerant structures.  Lastly, we believe it would be productive to explore further the relationship between self-assembling structures and the placing of geometric constraints on local and self-stabilizing algorithms.
\section*{Acknowledgments}
I am grateful to Soma Chaudhuri, Christoph Lentzen, Jack Lutz, Paul Rothemund, David Soloveichik, Jukka Suomela and Erik Winfree for helpful discussions.
\bibliography{tile_references}
\end{document}